\documentclass[11pt]{article}
\usepackage[compact]{titlesec}
\titlespacing{\section}{0pt}{*1.0}{*1.0}
\titlespacing{\subsection}{0pt}{*1.0}{*1.0}
\titlespacing{\subsubsection}{0pt}{*1.0}{*1.0}

\usepackage[compact]{titlesec}
\usepackage{times}
\usepackage{latexsym}
\usepackage{amsfonts,amsthm,amssymb}
\usepackage{amsmath}
\usepackage{euscript}
\usepackage{amstext}
\usepackage{graphicx}
\usepackage{color}
\usepackage{url,hyperref}
\usepackage{verbatim}
\usepackage{xspace}
\usepackage{framed}
\usepackage[boxed]{algorithm}
\usepackage{algorithmic}
\usepackage{enumitem}

\newtheorem{lemma}{Lemma}[section]
\newtheorem{theorem}[lemma]{Theorem}

\newtheorem{definition}[lemma]{Definition}
\newtheorem{corollary}[lemma]{Corollary}
\newtheorem{proposition}[lemma]{Proposition}

\newenvironment{proofof}[1]{\smallskip\noindent{\bf Proof of #1}}%
        {\hspace*{\fill}$\Box$\par}

\newcommand{\eps}{\epsilon}

\newcommand{\etal}{et al.\xspace}

\newcommand{\cA}{{\cal A}}

\newcommand{\cP}{{\cal P}}

\newcommand{\cS}{{\cal S}}

\newcommand{\initOneLiners}{%
    \setlength{\itemsep}{0pt}
    \setlength{\parsep }{0pt}
    \setlength{\topsep }{0pt}
}

\newcounter{property}

\usepackage{tikz}
\usepackage{float}

\newcommand{\cC}{{\cal C}}
\newcommand{\cU}{{
\cal U}}

\begin{document}

%\title{New Approximation Algorithms for Non-Preemptive Real Time Scheduling}
%\title{New Approximation Algorithms for Non-Preemptive Throughput Maximization}
\title{A Conditional Lower Bound on Graph Connectivity\\ in MapReduce }

\author{
Sungjin Im\thanks{ Electrical Engineering and Computer Science, University of California. {\tt sim3@ucmerced.edu}. Supported in part by NSF grant CCF-1409130.} 
\and 
Benjamin Moseley\thanks{Tepper School of Business.  Carnegie Mellon Universiy. {\tt moseleyb@andrew.cmu.edu}.}
} 

\date{June 1, 2015}

\iffalse
\numberofauthors{2}
\author{
\alignauthor    
Sungjin Im \\
%\affaddr{Dept. of Electrical Engineering and Computer Science}\\
\affaddr{University of California - Merced}\\
\affaddr{Merced, CA 95343}\\
\email{sim3@ucmerced.edu}
\alignauthor
Benjamin Moseley \\
\affaddr{Washington University in St. Louis.}\\
\affaddr{St. Louis, MO 63130}\\
\email{bmoseley@wustl.edu}
}
\fi

%\date{}

\maketitle

\begin{abstract}
MapReduce (and its open source implementation Hadoop) has become the \emph{de facto} platform for processing large data sets. MapReduce offers a streamlined computational framework by interleaving sequential and parallel computation while hiding  underlying system issues from the programmer. Due to the popularity of MapReduce, there has been attempts in the theoretical computer science community to understand the power and limitations of the MapReduce framework.  In the most widely studied MapReduce models each machine has  memory sub-linear in the input size to the problem, hence cannot see the entire input. This restriction places many limitations on algorithms that can be developed for the model; however, the current understanding of these restrictions is still limited.

In this paper, our goal is to work towards understanding problems which do not admit efficient algorithms in the MapReduce model.  We study the basic question of determining if a graph is connected or not.  We concentrate on instances of this problem where an algorithm is to determine if a graph consists of a single cycle or two disconnected cycles.  In this problem, locally every part of graph is similar and the goal is to determine the global structure of the graph.   We consider a natural class of algorithms that can store/process/transfer the information only in the form of paths, and show that no randomized algorithm cannot answer the decision question in a sub-logarithmic number of rounds. Currently, there are no absolute super constant lower bounds on the number of rounds known for any problem in MapReduce. We introduce some of the first lower bounds for a natural graph problem, albeit for a restricted class of algorithms. We believe our result makes a progress towards understanding the limitations of MapReduce.
\end{abstract}

\thispagestyle{empty}
\newpage
\setcounter{page}{1}

\newcommand{\kv}[2]{\langle #1 ; #2 \rangle}
\section{Introduction}
Today as data sets grow larger, practitioners have turned to utilizing distributed computing to processes these data sets and MapReduce and its open source implementation Hadoop have emerged as a standard distributed computational framework.  This framework has been adopted by hundreds of companies and universities.   The main success of MapReduce is owed to a simple framework offered to programmers which hides underlying system issues, such as fault tolerance and network communication, and lets the programmer focus on algorithms design.

A MapReduce computation informally works as follows and more formal specifics are given in Section~\ref{sec:prelim}. There is a set of machines available for data processing and the data is assumed to be arbitrarily partitioned across the machines at the beginning of the computation. Computation is performance in rounds and each round consists of a Map phase and a Reduce phase.   During the Map phase, a machine decides for each piece of local data stored on the machine, which machine(s) the data should be sent to.  After the Map phase occurs, the data is routed to the machines as specified by the mappers where it is then processed in the Reduce phase.  In the Reduce phase, each machine runs an algorithm over the data it is assigned and possibly creates new data that is to be mapped in a subsequent map phase.  The key requirements are that each machine can only process data that locally resides on machine during the reduce phase and can only decide how to map the data to other machines based on local information.

Due to the popularity of MapReduce, the theoretical computer science community has focused on developing a theoretical foundation for the framework \cite{jacob2014complexity,ChierichettiDK14,KumarMVV13,LattanziMSV11,BahmaniMVKV12,EneIM11,ChierichettiKT10,BahmaniKV12,SuriV11,KarloffSV10,MirzasoleimanKSK13,BroderPJVV14,AndoniNOY14,FeldmanMSSS08,GoodrichSZ11,BeameKS13}.  Several prior works have designed theoretical models of MapReduce \cite{FeldmanMSSS08,KarloffSV10,GoodrichSZ11,GoelM12,AndoniNOY14} and some of these are special cases of more general parallel computing models \cite{CullerKPSSSSE93,Valiant90}.  The models  introduced for MapReduce are similar.  Currently, the most well studied model of MapReduce was given in \cite{KarloffSV10}.  The popularity of the model comes from its simplicity and the abundant examples where theoretical algorithms developed for the model translate into efficient practical algorithms. For example, \cite{BahmaniMVKV12,EneIM11,BahmaniKV12}.  It is interesting to note, that while the model of \cite{KarloffSV10} was designed for MapReduce, it captures central constraints seen in distributed computation in general. 

The model of \cite{KarloffSV10} is as follows.  When a problem has size $n$ then the algorithm designer is given $n^{1-\eps}$ machines with $n^{1-\eps}$ memory where $\eps >0$ is a fixed constant.  The constant $\eps$ can be chosen by the algorithm designer and, so, for an algorithm to belong to the model, it is just required to  use this many machines and memory for \emph{some} constant $\eps$.   The main idea behind this restriction  is that that since MapReduce is used to process large data sets, the number of machines and memory on the machines should be sub-linear in the size of the problem.   It is assumed that computation on each of the machines takes $O(1)$ time during a Map or Reduce phase (even if it actually runs a polynomial time algorithm).  The goal is to design an algorithm which minimizes the number of rounds.  By minimizing the number of rounds, the algorithm minimizes the number of times off site network communication occurs, which is typically the most time consuming part of a MapReduce computation.  We note that with these constraints, an algorithm which runs in  $O(\textrm{polylog} n)$ rounds will never see the entire input to the problem.  Thus, to solve a problem in MapReduce, one must find a highly parallelized algorithm.    An ideal algorithm runs in $O(1)$ rounds and many fundamental problems admit such an algorithm \cite{KumarMVV13,AndoniNOY14,LattanziMSV11}.   Refinements have been suggested for this model such as including the total network communication and restricting the number of machines so that the total aggregate memory available is at most $\tilde{\Theta}(n)$ \cite{AndoniNOY14,GoodrichSZ11}. 

One  major success of MapReduce in both theory and practice is in the analysis of graphs.  One of the first problems which was shown to have a theoretically efficient MapReduce algorithm is the Minimum Spanning Tree problem \cite{KarloffSV10} and later algorithms were given for computing maximal and approximate weighted matchings in MapReduce \cite{LattanziMSV11,KumarMVV13}.  These algorithms run in $O(1)$ rounds, but they unfortunately have an additional requirement on the input.  It is assumed in both of these works that if $m$ is the number of nodes in the graph then the number of edges in the graph is at least $m^{1+c}$ for some constant $c >0$.  When assuming this, the input size is then at least $n = m^{1+c}$ and by setting $\eps = c/2$ one can then assume in the MapReduce model that each machine has sufficient memory to store all of the nodes of the graph, but not all of the edges.

It has been observed in practice that there are cases where large graphs are dense \cite{LeskovecKF05}.  However, it is typically the case that large graphs are sparse.  Unfortunately, the analysis of sparse graphs in MapReduce is currently not well understood.  Techniques that run in $O(1)$ rounds for dense graphs, do not lend themselves to sparse graphs.  There are techniques for simulating PRAM algorithms in the MapReduce setting \cite{KarloffSV10,GoodrichSZ11}, but these results only lead to $\Omega( \log n)$ round algorithms for sparse graphs.  Without simulation of a PRAM algorithm, an explicit $O(\log n)$ round algorithm is known to compute $s,t$-connectivity in sparse graphs \cite{KarloffSV10}.   A central open problem in the MapReduce setting, is resolving whether or not there exists $o(\log n)$ round algorithms for graph problems in MapReduce.  

Due to the resistance to developing such algorithms, it has been suggested that there may be strong lower bounds on the number of MapReduce rounds required for computing even simple problems such as determining if a graph is connected or not.   Unfortunately, proving strong lower bounds for any problem is also a central open problem in the area.  Fundamentally, developing lower bounds will help to separate what types of problem can be efficient solved in MapReduce and which require different computation frameworks to solve when data sets are large.  As mentioned, in essence the theoretical MapReduce framework is an abstraction of understanding what types of problems can be solved in parallel when machines have sub-linear local memory. Such an abstraction will likely continue to be useful beyond MapReduce as it simply captures challenges faced in many forms of distributed computation in general.

The major challenge in showing a lower bound for a MapReduce algorithm is that related lower bound techniques from communication complexity are challenging to apply.  If one is to show a lower bound in MapReduce, one should show that no machine learns large amount of information in a round.  However, while intuitively this is true, it is difficult to formally argue because a super-linear amount of information is communicated between the machines every round.  Thus, it is hard to determine what the machines learn overall.

\medskip
\noindent \textbf{Previous Work:} There have been some attempts to develop lower bounds for MapReduce. The work of \cite{AfratiSSU2013}  shows some problems which cannot be solved in 1 round and \cite{fish2014computational} gives lower bounds but does not allow a machine to run a polynomial time computation during a round. 
Jacob \etal \cite{jacob2014complexity} have shown an interesting game between the algorithm and adversary of connecting two given points requires $\Omega(\log n)$ rounds for a certain class of algorithms. However, the connectivity information is revealed to the algorithm in a restrictive manner which is part of game rule, and it is not clear how the game can be translated into a natural graph problem.

Recently, a novel result in the work of \cite{BeameKS13}  \footnote{Also see \cite{BeameKS13arxiv} for the full version.} has shown that $s,t$-connectivity requires $\Omega(\log n)$ rounds for a class of algorithms, but assumes that the algorithm must output the path connecting $s$ and $t$. Such an assumption is natural in the context of databases (which was the main setting \cite{BeameKS13} considered) where a full discovered relationship between $s$ and $t$ must be output explicitly. Later, this result was nicely extended in \cite{BeameKS13arxiv} to the problem of labeling each point with the ID of the connected component the point belongs to.

%We note that this result strongly uses the fact that the entire $s,t$ path must be output and this model does not capture the MapReduce connectivity algorithm of \cite{KarloffSV10}. 
Hence it has remained open if there are super constant lower bounds even in restricted MapReduce  models for decision problems arising in graphs. 

\medskip
\noindent \textbf{Results:}  In this work we study lower bounds in the MapReduce setting.  We study the decision problem of determining if a given graph is connected or not.   To study this problem, we focus on the following instance problem instance.  The given graph $G= (V,E)$ is an undirected graph which consists of either one cycle on $|V|$ nodes or two cycles each on $|V|/2$ nodes.  The goal is for the algorithm to determine if the graph is one or two cycles, thus determining if it is connected or not.  Note that in either of these cases, the total size of the input is $O(|V|)$ and in the MapReduce model, this enforces that each machine has memory at most $O(|V|^{1-\eps})$ for some constant $\eps >0$.

The lower bound we show is for certain classes of algorithms and specifics are given in Section~\ref{sec:restrictions}.  Roughly, we assume that the algorithm is required to store all information in the form of subpaths of the cycle(s).  Anytime the algorithm has two paths on a machine which intersect (or touch) each other, the algorithm can merge these paths into a single path.  Upon merging these paths, we assume that no matter the length of the path being stored, the algorithm is allowed to use $O(1)$ space to store each path, including the whole information regarding the path (eg. from which subpaths the path was merged).  Then, the algorithm can map each path using only the  information of the path itself, but not other disjoint paths. The goal for the algorithm is to find a path longer than $|V|/2$, certifying that the graph does not consist of one cycle or show that no such path exists, certifying that the graph consists of two cycles.  Using these restrictions, we give the following theorem.

\begin{theorem}
\label{thm:main}
No randomized algorithms, which store information only in the form of paths, merge intersecting paths to obtain longer paths, and map each path based only on the path itself (see Section~\ref{sec:restrictions} for full details on the assumptions), can decide if a given graph of $n$ points consists of one cycle or two disconnected cycles in $o(\log n)$ rounds in MapReduce. 
\end{theorem}

As mentioned, this is the first lower bound for a graph decision problem in the MapReduce setting.  Our model captures most known algorithms for graph connectivity in MapReduce such as the connectivity algorithm of \cite{KarloffSV10}.  We further note that while we specifically consider the model of \cite{KarloffSV10} the further refinements of the model given in \cite{AndoniNOY14,GoodrichSZ11} are more restrictive and our lower bounds hold there as well. Our proof techniques do not use previous literature on lower bounds in alternative models, but rather we argue directly about the information that can be acquired by an algorithm.  A key difference between our work and that of \cite{BeameKS13} is that we allow learned information about paths to be compactly represented in $O(1)$ space, giving a large amount of power to the algorithm. This additional power justifies that the algorithm merges two paths if they intersect since it needs no extra memory, and makes the decision problem more clear. We feel that our assumptions give a considerably more flexible framework than \cite{BeameKS13} since most of key algorithmic developments in MapReduce require compact representation of data. For example graph filtering \cite{KarloffSV10,LattanziMSV11} and coresets for clustering \cite{KumarMVV13}. Indeed, being able to compactly represent large amounts of information is at the core of most MapReduce algorithms and capturing this in the lower bound is key to understanding the limits of this technique.

We were unable to determine a lower bound for all algorithms. However, we believe that our work is a fundamental step in determining unrestrictive lower bounds for problems in MapReduce and gives an explanation on why $O(1)$ round algorithms  for sparse graph analysis have been challenging to develop.

\medskip
\noindent \textbf{Overview of our Analysis:} Intuitively,  whether the graph consists of one cycle or two,  the graph looks similar to the algorithm, which only can use partial data on each machine for local Reduce computation. Connected components/paths observed on each machine look similar to the algorithm, making it hard to distinguish between both cases. Indeed, if points are arbitrarily stored across machines in the beginning (but each machine can store at most $n^{1- \eps}$ points), it is easy to see a random permutation of which nodes are connected induces paths of length at most $O(1)$ on all machines with high probability. While it is tempting to apply this observation over rounds to show a super constant lower bound on the number of rounds required for the algorithm to find a path of length $n/2$, the main technical challenges come from dependency between rounds -- we already used full randomness of permutations to fail the algorithm in the first round.  Now the algorithm has the freedom to place the points based on the revealed random decision in the next round and we now have to argue about what the algorithm might have learned in the next round \emph{without} changing which points are connected again.

To address these dependency issues, we make a crucial observation which we call the Local Invariance Lemma (Lemma~\ref{lem:invariant}). This lemma allows us to continue this type of analysis by recursively finding events conditioned on which we can enjoy sufficient randomness over a subset of instances consistent with what the algorithm has learned so far. At a high-level, the lemma states that if two different graphs (cycle labelings) have a long path $H$ in common on the cycle(s), then each subpath of $H$ sufficiently far from the boundaries (endpoints of the path) is present on the same subset of machines for many rounds, no matter which cycle (labeling) is input to the algorithm. 

To use the Invariance lemma, we refine the process of sampling a random permutation on a cycle into multiple levels where we create a collection of larger ordered sets of points in each level/round while preserving the order of points we fixed in the previous levels. 
This refined process can be viewed as a tree (see Figure~\ref{fig:perm-tree} for an illustration) where each leaf node is an instance. Here each ordered set is nothing but a subpath which will appear in each permutation of ordered sets in $S$; here each permutation corresponds to an instance. Let's call $S$ as a partition and let $\cP(S)$ denote all instances generated from $S$.  We show that in one additional round, the algorithm can increase the length of discovered paths only by a constant factor. Namely, conditioned on paths being of length at most $L$ in  some  partition $S$, the paths found in the next round from $S$ have length $O(L)$. This is where the Local Invariance Lemma kicks in. The lemma ensures that all ordered sets/paths in all instances in  $\cP(S)$ are small/short until the current round and allows us to identify 
a set of subpaths, each of which is present on the same set of machines for all instances in  $\cP(S)$.  In this sense, we can say the partition $S$ is \emph{veiled} (hidden) from the algorithm. Then,  we can use randomness over the permutations/instances in  $\cP(S)$ with respect to the identified subpaths in order to show that the algorithm can increase the length of learned paths by at most a constant factor with high probability. We now use the Invariance lemma again to show that each child (one more level refined partition) of $S$ is still veiled from the algorithm with high probability.
Our analysis shows that there are sufficiently many veiled  paths from the root to the leaf nodes in the tree, meaning that the algorithm  makes a small progress in every round in most cases.

\section{Preliminaries}
\label{sec:prelim}

It is assumed in MapReduce that data is represented at $\kv{key}{value}$ tuples.  Here the $key$ and $value$ are binary strings.  Intuitively, the $value$ is a piece of information and the $key$ is an address of a virtual machine.   The Map phase is defined as a function $f$ which takes a \emph{single} $\kv{key}{value}$ tuple and outputs a set of $\kv{key}{value}$ tuples.  The map function could be to be deterministic or randomized and is implemented by the algorithm designer.  The mapping operation is stateless and takes a single tuple so that the $\kv{key}{value}$ tuples can be mapped easily in a distributed fashion.  Intuitively, the map function sets the keys of the tuples to specify which machine the value should be sent to.  Note this the map function could possible create many tuples with the same value, thereby sending the value to several machines.  

After the map phase, a shuffle phase occurs which is done automatically by the system.  The shuffle phase routes the data with the same key to the same virtual machine.  The shuffle phase is typically the most time consuming portion of the round as this is when the massive data is routed through the network between the machines.

In the reduce phase, a reduce function $g$ takes as input all $\kv{key}{value}$ tuples with the same key $k$.  The reduce function is specified  by the algorithm designer and generally it is during the reduce phase where non-trivial computation is performed on the data. The reason non-trivial computation can be performed is because in this phase the algorithm designer knows the data it is operating over since the tuples mapped to the function was specified by the implemented map function. Like the map function, the reduce function can be deterministic or randomized. The output of the reduce function $g$ is a set of  $\kv{key}{value}$  pairs all with the same key. The output from the reduce phase is either more tuples that are to be mapped in a subsequent map phase or the final output. Parallelism is exploited in the reduce phase because tuples with different keys can be reduced by different machines.

Throughout this paper, we will assume the model of \cite{KarloffSV10} where we assume there are $n^{1-\eps}$ machines with $n^{1-\eps}$ memory where $\eps >0$ is a constant specified by the algorithm designer.  During a round, we assume that the computation performed by the map or reduce functions are allowed to be unbounded time functions (although the model requires they be polynomial time functions).  Our goal is to bound the number of rounds of communication required by the algorithm.

\section{Lower Bound}
\subsection{Restrictions on the Algorithms} 
\label{sec:restrictions}

There are $n+1$ points with unique IDs that are either on one cycle or two disconnected cycles, which is the question the algorithm is challenged to answer. There are $n^{1- \eps}$ machines, each with $n^{1 - \eps}$ units of memory where $0< \eps < 1$ is a fixed constant. A path $P$ is a consecutive sequence of points on the cycle(s) in clockwise order. The length of $P$ is defined as the number of vertices on $P$ and denoted as $|P|$. A path could be of length 1. All information is stored/processed/transferred in the form of paths. In particular, at any point in time, either before or after Reduce, the algorithm has a collection of subpaths of the cycle(s) on each machine.

We let $\cA^r_m(I)$ denote the collection of paths that the algorithm $\cA$ has on machine $m$ for instance $I$ at the beginning of the Reduce operation in round $r$. Since machines don't communicate with one another during Reduce, to emphasize machines being independent, we may say that machine $m$ has paths $\cA^r_m(I)$ for instance $I$ before Reduce in round $r$. When the instance $I$ is clear from the context, it can be omitted. Then, during Reduce each machine merges two intersecting subpaths repeatedly until all paths it has are pairwise non-intersecting/disjoint. Here we say that two paths are intersecting if the union of points on the two paths forms a path on a cycle. Otherwise, we say they are non-intersecting or disjoint. When we say two paths are properly intersecting, we mean that the two have a common point. Let $\cA^{r,+}_m$ denote the collection of \emph{disjoint} subpaths the algorithm $\cA$ has on machine $m$ \emph{after} Reduce in round $r$. Note that initially the $n+1$ points are stored across machines, so $\cA^1_m$ is a subset of points (or equivalently paths of length 1). A point may be present on multiple machines at the beginning.

%\vspace{-.3cm}
\begin{center}
\begin{tabular}[r]{|c|}
\hline
\\
\begin{minipage}{14.9cm}
\hspace{0.5cm} \textbf{Summary of the Algorithms' Restrictions} 
\begin{enumerate}
\setlength{\itemsep}{0pt}
\setlength{\parskip}{0pt}

\item Each machine can store its current information only in the form of a collection of paths.
\item Each machine can combine paths only when they are intersecting. 
\item Each machine can decide where to send a path $P$ it has based only on the path $P$, the current round number,  the machine index, and the history of the subpaths in $P$. 
\item The algorithm must have a path of length at least $n/2$ on a machine to decide  if the graph consists of a single cycle or not. 
\end{enumerate}
\end{minipage}\\\\
\hline
\end{tabular}
\end{center}

We assume that the algorithm can store each path using only \emph{one} unit of memory. Further, we allow a path $P$ to store all `local' history pertaining to itself -- not only all points on the path but also how the path was obtained, and even the history of its individual subpaths.  By history, we mean how the subpaths were obtained by the algorithm. However, the path $P$ carries no information on points on other disjoint paths. After obtaining disjoint paths, the algorithm (each machine $m$) decides where to send each path $P \in \cA^{r,+}_m$, based solely on the path's local information, the current round, and the machine index $m$. A path can be sent to multiple machines. Here the only constraint is the limited memory capacity of individual machines -- the algorithm must satisfy $|\cA^{r}_m| \leq n^{1-\eps}$ for all $r$ and $m$. Since all information is stored only in the form of paths, the only way for the algorithm to decide if the graph consists of a single cycle or not is to have a path of length at least $n/2$ on a machine.

\subsection{Analysis}
We fix a deterministic algorithm $\cA$ under the restrictions discussed in the previous section. Let $\rho := 1024 / \eps$ and $R:= (\eps /2) \log_{\rho} n$. For simplicity, we assume that $1 / \eps$ and $R$ are integers. We will show the following main lemma.

\begin{lemma}
\label{lem:main}
	Consider a random instance $I$ where $n+1$ points are randomly placed on a cycle. For all $1 \leq r \leq R$ and machines $m$, $\cA^{r,+}_m(I)$ has no path with length $2 \rho^r$ with high probability. 
\end{lemma}

This lemma implies that no deterministic algorithm can find a path of length $n/2$ in $o(\log n)$ rounds regardless of whether the graph has one cycle or two, hence proving Theorem~\ref{thm:main} by Yao's Min-Max Principle.

\subsubsection{Local Invariance Lemma}
In this section we prove a very useful lemma which we call the Local Invariance Lemma. Roughly speaking, the lemma states that if two instances have a long common subpath $H$, then the algorithm behaves very similarly for both instances with respect to the common subpath, meaning that the algorithm has most of $H$'s subpaths on the same subsets of machines for many rounds. For a path $H$, we say $H'$ is $t$-trimmed subpath of $H$ if $H'$ is obtained by removing $t$ points from both ends of $H$. Recall that we say that two paths intersect if the union of the points on the two is a consecutive sequence of points on the cycle. For two paths $P'$ and $P$, let $P' \wedge P$ denote the longest common subpath of $P'$ and $P$. For a set of paths $T$ and a path $P$, let $T \wedge P := \{P' \wedge P \; | \; P' \in T\}$. Also we let $T | P$ denote the subset of paths in $T$ that are subpaths of $P$. Note that $(T | P) \wedge P = T | P$ and $(T \wedge P) | P = T \wedge P$. 

\begin{lemma}[Local Invariance Lemma]
\label{lem:invariant}
Let $I_1$ and $I_2$ be instances that have an identical path $H$. Let $H_r$ denote the $\rho^r / 32$-trimmed subpath of $H$. Suppose that for all $1 \leq r \leq R$ and machines $m$, all paths in $\cA^{r,+}_m(I_1) | H_r$ are of length at most $\rho^r / 8$, and all paths $(\cA^{r,+}_m(I_1) \setminus ( \cA^{r,+}_m(I_1) | H_r))  \wedge H_r$ are of length at most $\rho^r / 32$. Then it follows that 

\begin{itemize}
\setlength{\itemsep}{0pt}
\setlength{\parskip}{0pt}
\item Claim (i): all paths in $\cA^{r,+}_m(I_2) | H_r$ are of length at most $\rho^r / 8$, and all paths $(\cA^{r,+}_m(I_2) \setminus ( \cA^{r,+}_m(I_2) | H_r))  \wedge H_r$ are of length at most $\rho^r / 32$.; and 
\item Claim (ii): $\cA^{r,+}_m(I_1) | H_r = \cA^{r,+}_m(I_2) | H_r$. 
\end{itemize}
\end{lemma}

Note that in the lemma, $(\cA^{r,+}_m(I_1) \setminus ( \cA^{r,+}_m(I_1) | H_r))$ refers to the paths in $\cA^{r,+}_m(I_1)$ that are not fully contained in $H_r$. 
Hence, what all paths $(\cA^{r,+}_m(I_1) \setminus ( \cA^{r,+}_m(I_1) | H_r))  \wedge H_r$ being of length at most $\rho^r / 32$ means  is that all paths in $\cA^{r,+}_m(I_1)$ that are not fully contained in $H_r$ intersect $H_r$ by at most $\rho^r / 32$. The first part of Claim (i) is redundant since it follows from Claim (ii). However, we keep it explicit since it will be useful later for analysis. 

Before we prove the lemma we explain what it means in other words and how we can exploit it. We will show that the length of paths machines have can grow only by a constant factor in each round as claimed in Lemma~\ref{lem:main}. Intuitively, most subpaths of $H$ will move to the same machines for many rounds due to the `locality' assumptions imposed on the algorithm. However, the similarity will diminish in rounds since `boundary' subpaths of $H$ can be mapped to different machines affected by other paths that are not fully contained in $H$. This is why we trim out $H$ by an exponentially increasing margin on both ends, which is formalized in $H_r$. The first claim says that once we restrict $H$ to $H_r$, all subpaths the algorithm has found do not overlap with $H_r$ a lot. The second claim states that all subpaths fully contained in the trimmed subpath $H_r$ are stored on the same sets of machines for both instances $I_1$ and $I_2$. This will allow us to identify a large number of subpaths shared by many instances that the algorithm has to connect to obtain longer paths. By defining the probability subspace on the instances, we will be able to repeatedly apply our probabilistic argument over multiple rounds.

We now prove Lemma~\ref{lem:invariant} by induction on the number of rounds. Recall that $\cA^{r,+}_m(I)$ is the paths the machine $m$ has for instance $I$ after Reduce in round $r$.   

\begin{proofof}[Lemma~\ref{lem:invariant}]
We start with proving the base case. Recall that each point on $H$ is on the same subset of machines at the beginning, i.e. $\cA^1_m(I_1) \wedge H= \cA^1_m(I_2) \wedge H$ for all $m$; recall that $\cA^1_m$ only has paths of a unit length (or equivalently points). Since each machine $m$ only can merge paths on itself, for any $P^* \in \cA^{1,+}_m(I_2) \setminus (\cA^{1,+}_m(I_2)  | H_1) \wedge H$, all points (each point is a path of length 1) on $P^*$ are in $\cA^{1}_m(I_2) \wedge H= \cA^{1}_m(I_1) \wedge H$. Further,  $P^*$ is not a subpath of $H_1$. Hence $m$ can make a path $P$ that is not a subpath of $H_1$ in round 1 for $I_1$ such that $P \wedge H_1$ is as long as $P^* \wedge H_1$. This immediately proves the second part of Claim (i). As mentioned before, the first part of Claim (i) will follow from Claim (ii).

To see Claim (ii), consider any path $P$ in $\cA^{1,+}_m(I_1) | H_1$. Note that all points on $P$ are in $\cA^1_m(I_1)$, but neither of the two points adjacent to $P$; otherwise $\cA^1_m(I_1)$ will contain a path longer than $P$, but not $P$.  This is the same for instance $I_2$ since $\cA^1_m(I_1) \wedge H= \cA^1_m(I_2) \wedge H$. Hence $P \in \cA^{1,+}_m(I_2)$, and since $P$ is a subpath of $H_1$, we have that $\cA^{r,+}_m(I_1) | H_r \subseteq \cA^{r,+}_m(I_2) | H_r$. Similalry, we can show $\cA^{r,+}_m(I_2) | H_r \subseteq \cA^{r,+}_m(I_1) | H_r$, hence we have Claim (ii). (We note that trimming $H$ is crucial: it may happen that $\cA^{r,+}_m(I_1) | H \neq \cA^{r,+}_m(I_2) | H$.)

\smallskip
We prove the lemma holds for round $r$ assuming that it does for round $r -1$. Since each machine $m$ sends out each path in $\cA^{r-1,+}_m$ to machines as it is for round $r$, Claims for round $r-1$ imply:

\begin{proposition}
\label{prop:before-reduce-1}
All paths in $\cA^{r}_m(I_1) \wedge H_{r-1}$ or $\cA^{r}_m(I_2) \wedge H_{r-1}$ are of length at most $ \rho^{r-1} / 8$.
\end{proposition}

\begin{proposition}
\label{prop:before-reduce-2}
For all $m$, $\cA^{r}_m(I_1) | H_{r-1}= \cA^{r}_m(I_2) | H_{r-1}$. 
\end{proposition}

We begin with proving the second part of Claim (i) for round $r$. For the sake of contradiction, suppose $\cA^{r,+}_m(I_2) \setminus ( \cA^{r,+}_m(I_2) | H_r )$ has a path $P^*$ such that $P^* \wedge H_r$ has length greater than $\rho^r / 32 $. Let $P'$ be the subpath of $P^*$ that also fully contains $P^* \wedge H_r$ and has length one more than $P^* \wedge H_r$. 
If there exist $P_1, P_2, ..., P_k$ in $\cA^r_m(I_2)$ that are all subpaths of $H_{r-1}$ such that $P' \wedge H_r$ is a subpath of the concatenation of $P_1, P_2, ..., P_k$, then by Proposition~\ref{prop:before-reduce-2}, it must be that $P_1, P_2, ..., P_k \in \cA^r_m(I_1)$. Hence $\cA^{r,+}_m(I_1) | H_r$ must have a path as long as $P^* \wedge H_r$ that is not a subpath of $H_r$, which is a contradiction to the conditions of Lemma~\ref{lem:invariant}. Otherwise, there must exist $P_0$ in $\cA^r_m(I_2)$ including a point not in $H_{r-1}$ and  intersecting $H_{r}$, and other paths $P_1, P_2, ..., P_k$ in $\cA^r_m(I_2)$ that are all subpaths of $H_{r}$ such that $P' \wedge H_r$ is a subpath of the concatenation of $P_0, P_1, P_2, ..., P_k$. Thus, $P_0 \wedge H_{r-1}$ has length at least $(|H_{r-1}| - |H_r|) / 2 =  \rho^r / 32  -  \rho^{r-1} / 32  >  \rho^r/ 16$, which is a contradiction to Proposition~\ref{prop:before-reduce-1}, which restricts the length of $P_0 \wedge H_{r-1}$ to at most $\rho^{r-1} / 8$. As before, the first part of Claim (i) will follow from Claim (ii).

We now prove Claim (ii) for round $r$. For the sake of contradiction, suppose there is a path $P^* \in (\cA_m^{r,+}(I_1) | H_r) \setminus (\cA_m^{r,+}(I_2) | H_r)$; the other case can be  handled symmetrically. Say machine $m$ obtained $P^*$ by combining $P_1, P_2, ..., P_k$ in $\cA^{r}_m(I_1)$ in round $r$. Since $H_r$ is a subpath of $H_{r-1}$, by Proposition~\ref{prop:before-reduce-2}, $P_1, P_2, ..., P_k \in \cA^{r}_m(I_2)$. Now the only possible reason $P^* \notin \cA^{r,+}(I_2) | H_r$ is that there is a path $P_0 \in \cA^r_m(I_2)$ that intersects $P^*$ (thus intersecting $H_r$) so that machine $m$ obtains a longer path for $I_2$ than $P^*$ for $I_1$ in round $r$. If $P_0$ is a subpath of $H_{r-1}$, then $P_0 \in \cA_m^r(I_1)$ by Proposition~\ref{prop:before-reduce-2}. Hence $m$ obtains a longer path than $P^*$ for $I_1$, so $P^* \notin \cA_m^{r,+}(I_1)$, which is a contradiction. Thus $P_0$ must include a point not in $H_{r-1}$ and intersect $H_r$, which yields $|P_0 \wedge H_{r-1}| \geq (|H_{r-1}| - |H_r|) / 2 >  \rho^r/ 16$, which is a contradiction to Proposition~\ref{prop:before-reduce-1} as we argued above. 
This completes the proof of Lemma~\ref{lem:invariant}.
\end{proofof}

\subsubsection{Constructing a Random Lower Bound Instance}
In this section we construct our lower bound instance and prove Lemma~\ref{lem:main} and Theorem~\ref{thm:main}. As stated in Lemma~\ref{lem:main} our random instance is generated by placing $n+1$ points on a cycle in random order. While the instance itself is simple, our analysis is non-trivial and repeatedly uses Lemma~\ref{lem:invariant}. Towards this end, we will have to define the probability space carefully. 

The $n+1$ points are indexed by their IDs, $0, 1, ..., n$. We fix the $n+1$ locations where points can be placed and index them by $Q = \{q_0, q_1, q_2, ..., q_n\}$  in clockwise order. Here, $q_0$ is the `pivot', and is the next to $q_n$ on the cycle. For simplicity, assume that $n$ is a power of $\rho:= 1024 / \eps$. For the sake of analysis, we give an alternative but equivalent definition of a random permutation of the points on the cycle. In the probability space, an outcome can be described by a one-to-one mapping $f$ from point IDs $\{0\} \cup [n]$ to point locations $Q$ with the invariant $f(0) = q_0$ regarding the pivot point. We denote the entire outcome space as $\cU$. Then each possible $f$ corresponds to a unique permutation of $[n]$, so it follows that $|\cU| = n!$. 

Since we will repeatedly apply Lemma~\ref{lem:invariant} by conditioning on certain events, we gradually refine the outcome space in multiple levels. We enumerate all permutations of $[n]$ using the following rooted tree $\tau$. Partition $[n]$ into $n/ \rho$ \emph{disjoint} ordered sets of equal size $\rho$, which we call level-1 \emph{segments}. The family of such partitions, which we call level-1 partitions, is denoted as $\cS_1$. Level-1 partitions are the children of the root. Note that there are $|\cS_1| = n! / (n / \rho)!$ level-1 partitions. 
In general, for each level-$(r-1)$ partition $S$ in $\cS_{r-1}$, we partition level-$(r-1)$-segments in $S$ into $n / \rho^{r}$ ordered sets of equal size, which we call level-$r$ segments. 
The family of such partitions is denoted as $\cC_{r}(S)$ since they are all derived from a specific level-$(r-1)$ partition, $S$. We recurse until $r = \log_{\rho} n$. When $r = \log_{\rho} n$, each level-$r$ partition in $\cS_r$ consists of a single segment of length $n$, which is exactly a permutation of point IDs, $[n]$.  Hence we can get a uniform distribution over permutations of $[n]$ by sampling $S_1$ from $\cS_1$ uniformly at random, and sampling $S_2$ from $\cC(S_1)$, and so on. See Figure~\ref{fig:perm-tree} for an illustration of the tree $\tau$. 
Let $\cU_r(S)$ denote the outcome subspace where all outcomes in $\cU_r(S)$ are descendants of $S$ in the tree $\tau$. To save notation, we allow $\cU_r(S)$ to refer to the uniform distribution over the outcome subspace $\cU_r(S)$. Likewise we can view $\cS_r$ as a probability distribution. Define a level-$r$ arc to be a sequence of $\rho^r$ locations starting with $q_{1+ (k-1) \rho^r}$ and ending with $q_{k \rho^r}$ for an integer $k$. Note that for any level-$r$ partition $S \in \cS_r$ and for all $I \in \cU_r(S)$, every level-$r$ segment in $S$ is mapped to a unique level-$r$ arc. See Figure~\ref{fig:arc}.

%For notational consistency, let $\cS_0$ be the singleton set where the unique element is the set of all level-0 segments, i.e. $[n]$.

\newcommand{\lt}[1]{\langle #1 \rangle}
\begin{figure}[tp]
\begin{center}
\includegraphics[width=.6\textwidth]{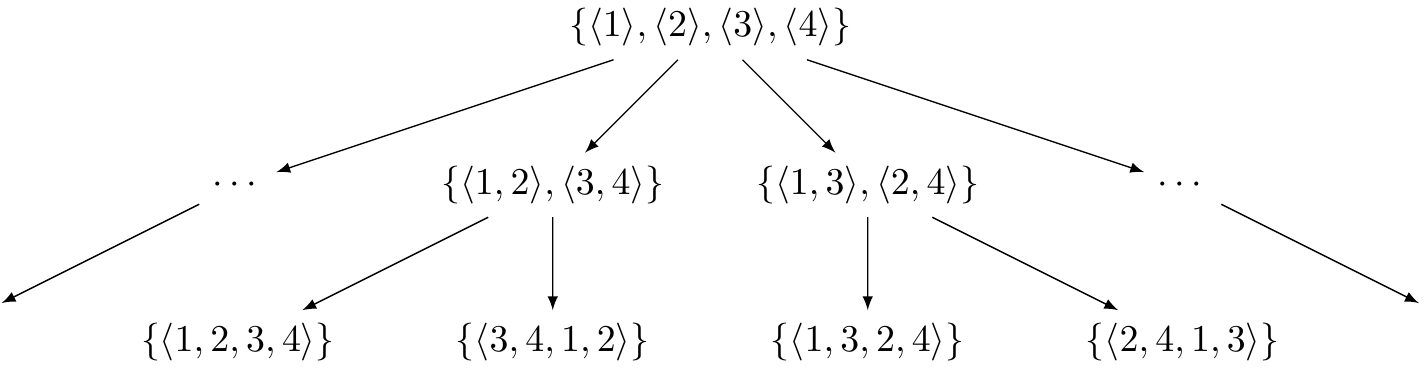}
\end{center}
\vspace{-1cm}
\caption{ \footnotesize{\textnormal{An illustration of permutation tree $\tau$ when $n = 4$ and $\rho =2$. The root has $4! / 2! = 12$ children, each of which has $2! = 2$ children. Each leaf node corresponds to a permutation of $[n]$. $\langle \cdots \rangle$ is used to denote an ordered set of point IDs.}}}
\label{fig:perm-tree}
\end{figure}
\iffalse%%%%%%%%%%%
\begin{tikzpicture}%
[
line/.style={>=latex,draw=black,fill=black, thin},
th/.style={>=latex,draw=black, thick},
rec1/.style={>=latex,draw=black,rounded corners, fill = blue!20, thin},
rec2/.style={>=latex,draw=black, rounded corners,fill = white!20, thin},
,scale = 0.8
]
\tikzstyle{vtx}=[circle,draw, fill = blue!20, inner sep=1pt];
\begin{scope}
\node (root) at (0, 0) {$\{ \lt{1}, \lt{2}, \lt{3}, \lt{4}\}$};
\node (a) at (-6, -2) {$ \cdots $}; 
\node (b) at (-2, -2) {$\{ \lt{1, 2}, \lt{3, 4}\}$};
\node (c) at (2, -2) {$\{ \lt{1, 3}, \lt{2, 4}\}$};
\node (d) at (6, -2) {$ \cdots $}; 
\draw[->,line] (root) -- (a) node[] {};
\draw[->,line] (root) -- (b) node[] {};
\draw[->,line] (root) -- (c) node[] {};
\draw[->,line] (root) -- (d) node[] {};
\node (b1) at (-6, -4) {$\{ \lt{1, 2, 3, 4}\}$};
\node (b2) at (-2, -4) {$\{ \lt{3, 4, 1, 2}\}$};
\node (c1) at (2, -4) {$\{ \lt{1, 3, 2, 4}\}$};
\node (c2) at (6, -4) {$\{ \lt{2, 4, 1, 3}\}$};
\draw[->,line] (a) -- (-9, -3.5) node[] {};
\draw[->,line] (b) -- (b1) node[] {};
\draw[->,line] (b) -- (b2) node[] {};
\draw[->,line] (c) -- (c1) node[] {};
\draw[->,line] (c) -- (c2) node[] {};
\draw[->,line] (d) -- (9, -3.5) node[] {};
\end{scope}
\end{tikzpicture}
\fi%%%%%%%%%%%

\begin{figure}[tp]
\begin{center}
\begin{tikzpicture}[scale = 0.5]
[
line/.style={>=latex,draw=black,fill=black, thin},
th/.style={>=latex,draw=black, thick},
rec1/.style={>=latex,draw=black,rounded corners, fill = blue!20, thin},
rec2/.style={>=latex,draw=black, rounded corners,fill = white!20, thin},
,scale = 0.8
]
\tikzstyle{vtx}=[circle,draw, fill = gray!20, inner sep=2pt];
\tikzstyle{id}=[circle,draw, fill = gray!20, inner sep=3pt];
\begin{scope}
\draw (0,0) circle (3cm);
\node (q0) at (3, 0) [vtx] {$q_0$};
\node (q1) at (0.927, -2.853) [vtx] {$q_1$};
\node (q2) at (-2.426, -1.765) [vtx] {$q_2$};
\node (q3) at (-2.426, 1.765) [vtx] {$q_3$};
\node (q4) at (0.927, 2.853) [vtx] {$q_4$};
\node (i0) at (4, 0) {$0$};
\node (i1) at (1.236, -3.804) {$3$};
\node (i2) at (-3.235, -2.351) {$4$};
\node (i3) at (-3.235, 2.351) {$1$};
\node (i4) at (1.236, 3.804) {$2$};
\node (l) at ( -7, 0) {$\lt{3,4, 1, 2}$:};
\end{scope}
\end{tikzpicture}
\vspace{-4mm}
\caption{ \footnotesize{\textnormal{The point IDs in the ordered set are mapped to $q_1, q_2, \cdots, q_n$ in the order they appear in the ordered set. When $n = 4$ and $\rho = 2$, level-1 arcs are $\lt{q_1, q_2}$ and $\lt{q_3, q_4}$. $\{ \lt{1, 2}, \lt{3,4}\}$ is a level-1 partition. For both instances $\lt{1,2, 3,4}$ and $\lt{3,4,1,2}$, which are children of $\{ \lt{1, 2}, \lt{3,4}\}$, level-1 segments $\lt{1, 2}$ and $\lt{3,4}$ are mapped to level-1 arcs $\lt{q_1, q_2}$ and $\lt{q_3, q_4}$. For instance $\lt{3,4,1,2}$, $\lt{3,4}$ is mapped to $\lt{q_1, q_2}$, and $\lt{1,2}$ to $\lt{q_3, q_4}$.}}}
%\vspace{-4mm}
\label{fig:arc}
\end{center}
\end{figure}
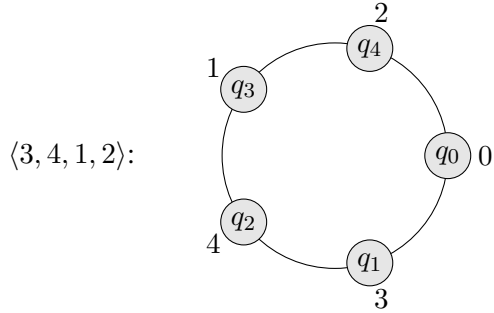

\vspace{.2cm}
\subsubsection{Main Analysis}

We begin our analysis with showing that w.h.p. the algorithm cannot obtain a path of length greater than $\rho / 32$ in round 1. 

\begin{lemma}
\label{lem:beginning}
Let $I$ be an instance sampled from $\cU$ uniformly at random. The probability that for all $m$, all paths on $\cA^{1,+}_m(I)$ have length at most $\rho / 32$ is at least $1 - 1/ n^2$. 
\end{lemma}
\begin{proof}
Consider any fixed machine $m$, which contains at most $n^{1- \eps}$ points (IDs) in the beginning, i.e. $|\cA^1_m| \leq n^{1-\eps}$. Consider any subset $Q'$ of $\kappa :=  \rho / 32$ consecutive point locations in $Q$ (not including $q_0$) on the cycle. The number of instances where all those locations in $Q'$ are mapped to points IDs in $\cA^1_m$ is at most $\left( \prod_{i = 0}^{\kappa - 1} (n^{1 - \eps} - i) \right) \cdot (n - \kappa)!$. Hence the probability that $m$ has a path over the locations on $Q'$ after Reduce in round 1 is at most $\left( \prod_{i = 0}^{\kappa - 1} (n^{1 - \eps} - i) \right) \cdot (n - \kappa)! / n! \leq (2 / n^\eps)^\kappa \leq 1 / n^4$ for a sufficiently large $n$; recall $\rho = 1024 / \eps$. Taking a union bound over all machines $m$ and all $Q'$ gives the lemma. 
\end{proof}

In the following, we formalize how much the algorithm learns about segments over rounds. Roughly speaking, a partition $S \in \cS_r$ is `veiled' from the algorihtm if it induces a outcome subspace $\cU_r(S)$ where all outcomes contain the same set of segments that the algorithm barely learned by round $r$. 

\newcommand{\veiled}{\textsf{Veiled}}
\begin{definition}
\label{def:veiled}
We say that a level-$r$ partition $S \in \cS_r$ is veiled and denote it as $\veiled_r(S)$ if for all machines $m$ and segments $s \in S$, 
\begin{enumerate}
\setlength{\itemsep}{0pt}
\setlength{\parskip}{0pt}
\item for all $I \in \cU_r(S)$, $\cA^{r,+}_m(I) | s_r$ has no path of length more than $\rho^r / 8$, and $(\cA^{r,+}_m(I) \setminus ( \cA^{r,+}_m(I) | s_r))  \wedge s_r$ has no path of length more than $\rho^r / 32$; and 
\item for all $I_1, I_2 \in \cU_r(S)$, $\cA^{r,+}_m(I_1) | s_r = \cA^{r,+}_m(I_2) | s_r$. 
\end{enumerate} 
\noindent where $s_r$ denotes the $\rho^r / 32$-trimmed subpath of $s$. 
\end{definition}

\begin{definition}
	\label{def:a-veiled}
	For a level-$r$ partition $S \in \cS_r$, we say that $S$ is ascendant-veiled if all ascendant partitions of $S$ in the permutation tree $\tau$ are veiled. 
\end{definition}

The Local Invariant Lemma (\ref{lem:invariant}) makes Definition~\ref{def:veiled} more applicable. 

\begin{lemma}
\label{corollary:veiled}
A level-$r$ partition $S \in \cS_r$ is veiled if 
	\begin{enumerate}
	\setlength{\itemsep}{0pt}
\setlength{\parskip}{0pt}
		\item $S$ is ascendant-veiled; and
		\item there exists an  instance $I \in \cU_r(S)$ such that for all machines $m$ and segments $s \in S$, $\cA^{r,+}_m(I) \wedge s_r$ has no path of length more than $ \rho^r / 32$. 
	\end{enumerate} 
\noindent where $s_r$ denotes the $\rho^r / 32$-trimmed subpath of $s$. 
\end{lemma}
\begin{proof}
 	We apply Lemma~\ref{lem:invariant} with $I_1 = I$ and $H = s$ for any level-$r$ segment $s$ in $S$, and with any $I_2 \in \cU_r(S)$. Note that both instances $I_1$ and $I_2$ have $s$ as a subpath. 	To show the lemma, we only need to show that for all $1 \leq r' \leq r$ and machines $m$, all paths in $\cA^{r,+}_m(I_1) | s_r$ are of length at most $\rho^r / 8$, and all paths in $(\cA^{r,+}_m(I_1) \setminus ( \cA^{r,+}_m(I_1) | s_r))  \wedge s_r$ are of length at most $\rho^r / 32$. This is immediate when $r' = r$ from the second condition stated in the lemma. 
 	
 	Next, consider any $r' < r$. Let $S'$ denote the unique level-$r'$ partition that is an ascendant of $S$ in $\tau$. Since $S$ is ascendant-veiled, by definition $S' \in \cS_{r'}$ is veiled. Note that the segment $s \in S$ is obtained by stitching $\rho^{r - r'}$ segments in $S'$ -- we denote the first and last as $a'$ and $b'$, respectively. Let $a'_{r'}$ and $b'_{r'}$ denote the $\rho^{r'}/32$-trimmed subpaths of $a'$ and $b'$, respectively. Note that $s_{r'}$ starts with $a'_{r'}$ and ends with $b'_{r'}$. First consider a path $P^* \in (\cA^{r,+}_m(I_1) \setminus ( \cA^{r,+}_m(I_1) | s_r))  \wedge s_r$. Observe that $P^*$ properly intersects $a'_{r'}$ (or $b'_{r'}$, resp.), but is not a subpath of $a'_{r'}$ (or $b'_{r'}$, resp.).  Since $S'$ is veiled, we know $P^*$ intersects $a'_{r'}$ (or $b'_{r'}$, resp.) by at most $\rho^{r'} / 32$, hence the second part of the claim holds. Finally, consider $P^* \in \cA^{r,+}_m(I_1) | s_r$. In this case, $s_r$ fully contains $P^*$. Let $S''$ denote the $\rho^{r'} / 32$-trimmed subpaths of segments in $S'$. If $P^*$ is fully contained in a trimmed subpath in $S''$, then we know that $P^*$ has length at most $\rho^{r'} / 8$ since $S'$ is veiled, so we are done. Otherwise, $P^*$ can properly intersect at most two trimmed subpaths in $S'$, each by at most $\rho^{r'} / 32$ where the two trimmed subpaths have $2* \rho^{r'} / 32$ points between them. Hence, $P^*$ can have length at most $\rho^{r'} / 8$. 
\end{proof}

By Lemma~\ref{lem:beginning} and Lemma~\ref{lem:invariant}, we show that most level-1 partitions are veiled. 
\begin{lemma}
\label{lem:veiled-beginning}
$\Pr_{S \sim \cS_1} [ \veiled_1(S)] \geq 1 - 1 / n^2$. 
\end{lemma}
\begin{proof}
Mark $S \in \cS_1$ if there is an instance $I \in \cU_1(S)$ such that for all $m$, all paths on $\cA^{1,+}_m(I)$ have length at most $ \rho  /32$. Note that $\{\cU_1(S)\}_{S \in \cS_1}$ partitions the outcome space $\cU$. By Lemma~\ref{lem:beginning}, a level-1 partition sampled from $\cS_1$ is marked with probability at least $1 - 1/ n^2$. We observe that a marked level-1 partition $S$ satisfies the condition stated in Lemma~\ref{corollary:veiled}; $S$ is a level-1 partition, so it is ascendant-veiled by Definition~\ref{def:a-veiled}.
\end{proof}

Due to space constraints, the proof of the following lemma can be found in the appendix

\begin{lemma}
\label{lem:no-long-from-veiled}
Consider any $r$ and veiled $S \in \cS_{r-1}$. Let $I$ be an instance sampled from $\cU_{r-1}(S)$. Then, the probability that for all machines $m$, all paths in $\cA^{r,+}_m(I)$ have length at most $ \rho^r /32$ is at least $1 - 1/ n^2$. 
\end{lemma}

\begin{lemma}
\label{lem:veiled}
	For any $2 \leq r \leq R$, veiled and ascendant-veiled partition $S \in \cS_{r-1}$, we have
$\Pr_{S' \sim \cC(S)} [ \veiled_r(S') ] \geq 1 - 1/ n^2$. 
\end{lemma}
\begin{proof}
The proof is similar to that of Lemma~\ref{lem:veiled-beginning}. Mark $S' \in \cC(S)$ if there is an instance $I \in \cU_{r}(S')$ such that for all machines $m$, all paths in $\cA^{r,+}_m(I)$ have length at most $\rho^r / 32$. By Lemma~\ref{lem:no-long-from-veiled}, a level-$r$-partition sampled from $\cC(S)$ is marked with probability at least $1 - 1/ n^2$. We observe that a marked level-$r$ partition $S$ satisfies the condition stated in Lemma~\ref{corollary:veiled}. 
\end{proof}

Using Lemma~\ref{lem:veiled-beginning} and \ref{lem:veiled}, we decondition the probability. 
\begin{corollary}
For all $1 \leq r \leq R$, $\Pr_{S \sim \cS_r} [ \veiled_r(S) ] \geq 1 - 1/ n$. 
\end{corollary}
Note that if $S \in \cS_r$ is veiled, then for all instances $I \in \cU_r(S)$, the algorithm has no path longer than $2\rho^r <  n / 2$ on any machines since the algorithm has no path
that fully contains a level-$r$ segment, and $R = (\eps / 2) \log_\rho p$. This proves Lemma~\ref{lem:main} and Theorem~\ref{thm:main}

\section*{Acknolwedgements}

The authors thank Paul Beame for kindly explaining the lower bound in \cite{BeameKS13} in detail.

%%%%%%%%%%%%%%%%%%%%%%%%%%%%%%%%%
\bibliographystyle{plain}
\bibliography{mrbib}

\begin{thebibliography}{10}

\bibitem{AndoniNOY14}
Alexandr Andoni, Aleksandar Nikolov, Krzysztof Onak, and Grigory Yaroslavtsev.
\newblock Parallel algorithms for geometric graph problems.
\newblock In {\em STOC}, 2014.

\bibitem{BahmaniKV12}
Bahman Bahmani, Ravi Kumar, and Sergei Vassilvitskii.
\newblock Densest subgraph in streaming and mapreduce.
\newblock {\em PVLDB}, 5(5):454--465, 2012.

\bibitem{BahmaniMVKV12}
Bahman Bahmani, Benjamin Moseley, Andrea Vattani, Ravi Kumar, and Sergei
  Vassilvitskii.
\newblock Scalable k-means++.
\newblock {\em PVLDB}, 5(7):622--633, 2012.

\bibitem{BeameKS13}
Paul Beame, Paraschos Koutris, and Dan Suciu.
\newblock Communication steps for parallel query processing.
\newblock In {\em Proceedings of the 32Nd Symposium on Principles of Database
  Systems}, PODS '13, pages 273--284, New York, NY, USA, 2013. ACM.

\bibitem{BeameKS13arxiv}
Paul Beame, Paraschos Koutris, and Dan Suciu.
\newblock Communication steps for parallel query processing.
\newblock {\em CoRR}, abs/1306.5972, 2013.

\bibitem{BroderPJVV14}
Andrei~Z. Broder, Lluis~Garcia Pueyo, Vanja Josifovski, Sergei Vassilvitskii,
  and Srihari Venkatesan.
\newblock Scalable k-means by ranked retrieval.
\newblock In {\em WSDM}, pages 233--242, 2014.

\bibitem{ChierichettiDK14}
Flavio Chierichetti, Nilesh~N. Dalvi, and Ravi Kumar.
\newblock Correlation clustering in mapreduce.
\newblock In {\em The 20th {ACM} {SIGKDD} International Conference on Knowledge
  Discovery and Data Mining, {KDD} '14, New York, NY, {USA} - August 24 - 27,
  2014}, pages 641--650, 2014.

\bibitem{ChierichettiKT10}
Flavio Chierichetti, Ravi Kumar, and Andrew Tomkins.
\newblock Max-cover in map-reduce.
\newblock In {\em WWW}, pages 231--240, 2010.

\bibitem{CullerKPSSSSE93}
David~E. Culler, Richard~M. Karp, David~A. Patterson, Abhijit Sahay, Klaus~E.
  Schauser, Eunice~E. Santos, Ramesh Subramonian, and Thorsten von Eicken.
\newblock Logp: Towards a realistic model of parallel computation.
\newblock In {\em PPOPP}, pages 1--12, 1993.

\bibitem{EneIM11}
Alina Ene, Sungjin Im, and Benjamin Moseley.
\newblock Fast clustering using mapreduce.
\newblock In {\em KDD}, pages 681--689, 2011.

\bibitem{FeldmanMSSS08}
Jon Feldman, S.~Muthukrishnan, Anastasios Sidiropoulos, Clifford Stein, and
  Zoya Svitkina.
\newblock On distributing symmetric streaming computations.
\newblock In {\em SODA}, pages 710--719, 2008.

\bibitem{fish2014computational}
Benjamin Fish, Jeremy Kun, {\'A}d{\'a}m~D{\'a}niel Lelkes, Lev Reyzin, and
  Gy{\"o}rgy Tur{\'a}n.
\newblock On the computational complexity of mapreduce.
\newblock {\em arXiv preprint arXiv:1410.0245}, 2014.

\bibitem{GoelM12}
A.~Goel and K.~Munagala.
\newblock Complexity measures for map-reduce, and comparison to parallel
  computing.
\newblock {\em Available at {\tt
  http://www.cs.duke.edu/\~{}kamesh/mapreduce.pdf}}, 2012.

\bibitem{GoodrichSZ11}
Michael~T. Goodrich, Nodari Sitchinava, and Qin Zhang.
\newblock Sorting, searching, and simulation in the mapreduce framework.
\newblock {\em CoRR}, abs/1101.1902, 2011.

\bibitem{jacob2014complexity}
Riko Jacob, Tobias Lieber, and Nodari Sitchinava.
\newblock On the complexity of list ranking in the parallel external memory
  model.
\newblock In {\em Mathematical Foundations of Computer Science 2014}, pages
  384--395. Springer, 2014.

\bibitem{KarloffSV10}
Howard~J. Karloff, Siddharth Suri, and Sergei Vassilvitskii.
\newblock A model of computation for mapreduce.
\newblock In {\em SODA}, pages 938--948, 2010.

\bibitem{KumarMVV13}
Ravi Kumar, Benjamin Moseley, Sergei Vassilvitskii, and Andrea Vattani.
\newblock Fast greedy algorithms in mapreduce and streaming.
\newblock In {\em SPAA}, pages 1--10, 2013.

\bibitem{LattanziMSV11}
Silvio Lattanzi, Benjamin Moseley, Siddharth Suri, and Sergei Vassilvitskii.
\newblock Filtering: a method for solving graph problems in mapreduce.
\newblock In {\em SPAA}, pages 85--94, 2011.

\bibitem{LeskovecKF05}
Jure Leskovec, Jon~M. Kleinberg, and Christos Faloutsos.
\newblock Graphs over time: densification laws, shrinking diameters and
  possible explanations.
\newblock In {\em KDD}, pages 177--187, 2005.

\bibitem{MirzasoleimanKSK13}
Baharan Mirzasoleiman, Amin Karbasi, Rik Sarkar, and Andreas Krause.
\newblock Distributed submodular maximization: Identifying representative
  elements in massive data.
\newblock In {\em NIPS}, pages 2049--2057, 2013.

\bibitem{AfratiSSU2013}
Anish~Das Sarma, Foto~N. Afrati, Semih Salihoglu, and Jeffrey~D. Ullman.
\newblock Upper and lower bounds on the cost of a map-reduce computation.
\newblock In {\em Proceedings of the 39th international conference on Very
  Large Data Bases}, PVLDB'13, pages 277--288. VLDB Endowment, 2013.

\bibitem{SuriV11}
Siddharth Suri and Sergei Vassilvitskii.
\newblock Counting triangles and the curse of the last reducer.
\newblock In {\em WWW}, pages 607--614, 2011.

\bibitem{Valiant90}
Leslie~G. Valiant.
\newblock A bridging model for parallel computation.
\newblock {\em Commun. ACM}, 33(8):103--111, 1990.

\end{thebibliography}

\appendix
\section{Omitted Proof}

\begin{proofof}[Lemma~\ref{lem:no-long-from-veiled}]
Fix an $r$, veiled $S \in \cS_{r-1}$, and machine $m$. Since $S$ is veiled we know that for any level-$(r-1)$ segment $s \in S$, for all instances in $\cU_{r-1}(S)$, all paths in $\cA^{r}_m$ overlaps with $s_{r-1}$ by at most $\rho^{r-1} / 8$; paths in $\cA^{r}_m$ only come from $\cA^{r-1,+}_{m'}$. Also note that $\cA^r_m | s_{r-1}$ is the same for all instances in $\cU_{r-1}(S)$. 

Say that the algorithm $\cA$ obtains a path of length more than $\rho^r / 32$ for instance $I$ that spans over a consecutive set $Q'$ of $\rho$ level-$(r -1)$ arcs of length $\rho^{r-1} / 32$. For notational simplicity, say $Q' = \{q_1, q_2, ..., q_\rho\}$. Say the level-$(r-1)$ segments of $I$ (or equivalently $S$), $t_1, t_2, ..., t_\rho$ are mapped to $q_1, q_2, ..., q_\rho$, respectively. 
Consider any such level-$(r-1)$ segment say $t_2$, except $t_1$ and $t_\rho$. We say that a path is a center path w.r.t. $t_2$ if the path is a subpath of $t_2$'s $\rho^{r-1} / 32$-trimmed path. Then, $\cA_{m}^r(I)$ must have a center path $P_2$ w.r.t. $t_2$. This is because without a center path, the arc $q_2$ only can be covered by at most $\rho^{r-1} / 32 +  \rho^{r-1} / 32$ on each end (the arc $q_2$ hosts $t_2$, and a non-center path w.r.t. $t_2$ can overlap the $\rho^{r-1} / 32$-trimmed subpath of $t_2$ by at most $\rho^{r-1} / 32$), but the arc has length $\rho^{r-1}$. Also we know that $\cA_{m}^r(I)$ has $P_2$ for all instances $I \in \cU_{r-1}(S)$ since $S$ is veiled and paths before Reduce in round $r$ come from those after Reduce in round $r-1$. Hence we have found center paths $P_2, ..., P_{\rho-1}$ that are in $\cA^r_m(I)$ for all instances $I \in \cU_{r-1}(S)$ that must be mapped to the level-$(r-1)$ arcs $q_2, q_3, ..., q_{\rho-1}$, respectively. 

Mark each level-$(r-1)$ segment $s \in S$ if $\cA^r_m(I)$ contains a center path w.r.t. $s$ for any $I \in \cU_{r-1}(S)$. Note that there are at most $n^{1-\eps}$ marked segments in $S$ since machine $m$ can hold at most $n^{1- \eps}$ paths. The above discussion shows that $\cA$ can obtain a path of length more than $\rho^r / 32$ on machine $m$ over $Q'$ only when the middle $\rho-2$ arcs, $q_2, ..., q_\rho$ are mapped to the marked segments.  The number of instances $I$ in $\cU(S)$ where every level-$(r-1)$ arc in $Q' \setminus \{q_1, q_\rho\}$ is mapped to a marked level-$(r-1)$ segment in $S$ is at most $\left( \prod_{i = 0}^{\rho - 3} (n^{1 - \eps}- i) \right) \cdot (n / \rho^{r-1} - (\rho-2))!$. 
Here we used the fact that there are $n / \rho^{r-1}$ level-$(r-1)$ arcs and at most $n^{1-\eps}$ marked segments in $S$. 

Hence the probability that $m$ has a path spanning over all level-$(r-1)$ arcs in $Q'$ is at most 
\begin{small}
$$\left( \prod_{i = 0}^{\rho - 3} (n^{1 - \eps}- i) \right) \cdot (n / \rho^{r-1} - (\rho-2))! / (n / \rho^{r-1})! \leq (\frac{2\rho^{r-1}}{n^\eps})^{\rho-2} \leq 1/ n^4.$$
\end{small}

The last inequality follows since $\rho =  1024 / \eps$ and $r < R = (\eps / 2) \log_\rho n$. Taking a union bound over all machines $m$ and all $Q'$ gives the lemma. 
\end{proofof}

\end{document}